\documentclass[pdflatex,sn-mathphys]{sn-jnl}% Math and Physical Sciences Reference Style
\usepackage{array}
\usepackage{multirow}
\usepackage{breakurl}
\usepackage{amsmath}
\usepackage{amsfonts}
\usepackage{amssymb}
\usepackage{amsthm}
\usepackage{graphicx}

\jyear{2022}%

\theoremstyle{thmstyleone}%
\newtheorem{theorem}{Theorem}

\newtheorem{corollary}{Corollary}
\newtheorem{conjecture}{Conjecture}

\newtheorem{problem}{Problem}

\theoremstyle{thmstyletwo}%
\theoremstyle{thmstylethree}%
\newtheorem{definition}{Definition}%

\begin{document}

\title[A Relative Church-Turing-Deutsch Thesis]{A Relative Church-Turing-Deutsch Thesis from Special Relativity and Undecidability}


\author*[1,2]{\fnm{Blake} \sur{Wilson}}\email{wilso692@purdue.edu}

\author[3]{\fnm{Ethan} \sur{Dickey}}\email{dickeye@purdue.edu}

\author[1,3]{\fnm{Vaishnavi} \sur{Iyer}}\email{iyer94@purdue.edu}

%\author[1,2]{\fnm{Vladimir} \sur{Shalaev}}\email{shalaev@purdue.edu}

%\author[1,2]{\fnm{Alexandra} \sur{Boltasseva}}\email{aeb@purdue.edu}

\author*[2,3,4]{\fnm{Sabre} \sur{Kais}}\email{kais@purdue.edu}

\affil[1]{\orgdiv{Elmore Family School of Electrical and Computer Engineering}, \orgname{Purdue University}, \orgaddress{\street{501 Northwestern Ave.}, \city{West Lafayette}, \postcode{47907}, \state{IN}, \country{USA}}}

\affil[2]{\orgdiv{Quantum Science Center}, \orgname{Oak Ridge National Laboratory}, \orgaddress{\street{1 Bethel Valley Road}, \city{Oak Ridge}, \postcode{37830}, \state{TN}, \country{USA}}}

\affil[3]{\orgdiv{Department of Computer Science}, \orgname{Purdue University}, \orgaddress{\street{305 N University St.}, \city{West Lafayette}, \postcode{47907}, \state{IN}, \country{USA}}}

\affil[4]{\orgdiv{Department of Chemistry}, \orgname{Purdue University}, \orgaddress{\street{560 Oval Drive}, \city{West Lafayette}, \postcode{47907}, \state{IN}, \country{USA}}}


\abstract{Beginning with Turing's seminal work in 1950, artificial intelligence proposes that consciousness can be simulated by a Turing machine. This implies a potential theory of everything where the universe is a simulation on a computer, which begs the question of whether we can prove we exist in a simulation. In this work, we construct a relative model of computation where a computable \textit{local} machine is simulated by a \textit{global}, classical Turing machine. We show that the problem of the local machine computing \textbf{simulation properties} of its global simulator is undecidable in the same sense as the Halting problem. Then, we show that computing the time, space, or error accumulated by the global simulator are simulation properties and therefore are undecidable. These simulation properties give rise to special relativistic effects in the relative model which we use to construct a relative Church-Turing-Deutsch thesis where a global, classical Turing machine computes quantum mechanics for a local machine with the same constant-time local computational complexity as experienced in our universe.}




\maketitle

\section{Introduction}\label{sec1}
%There is a story about a woodcarver who, after being heartbroken over her inability to have children, sculpted a young daughter out of wood. After years of laboring, her daughter's body was complete but she couldn't give her a soul. Desperate, the mother wished on a star for her daughter. That very night her daughter was born. Filled with glee, the mother cooked her daughter's first meal and introduced her to her toys. But, the daughter did not eat or play with her toys, and the village kids shunned her for being made of wood. Her mother assured her that she was just as real as the other village kids but she couldn't bring herself to believe it. Heartbroken over being an outcast, the daughter leaves home to prove to herself that she is a real person.

%Artificial intelligence is rapidly showing us that machines can mimic human behavior and may someday be indistinguishable from human-like consciousness. The daughter can't deny that she is made out of wood. Instead, she accepts that even if she isn't made out of flesh and blood, she is just as real. We are now facing the same question. 

Artificial intelligence (A.I.) is built on the assumption that mathematics can simulate brain behavior. Once we properly simulate a brain, many believe that the simulation will be conscious \cite{Turing1937OnEntscheidungsproblem, Arora2009ComputationalApproach}. Naturally, if consciousness is then computable, who is to say that our universe and all of our consciousness isn't \textbf{already} a simulation? This leaves open the \textbf{simulation hypothesis} that our universe is already a simulation on a computer and consciousness is a byproduct of that simulation \cite{Turing1950ComputingTuring, Bostrom2003AreSimulation, Deutsch1985QUANTUMCOMPUTER.}. In this paper, we study two questions at the heart of the simulation hypothesis, i) can a simulation decide if it is a simulation and ii) is quantum mechanics compatible with a classical simulation hypothesis? We answer both questions using a relative model of computation where a classical, global Turing machine simulates an arbitrary, local machine (e.g., artificial intelligence or the physical laws of the universe) encoded in its memory with any computational ability it desires, so long as the local machine is computable by a Turing machine. We show that special relativistic effects like time dilation and length contraction \cite{Griffiths2005Electrodynamics} exist in the relative model and we prove that the local machine cannot compute the global space or time resources nor simulation error incurred by the global machine. Then, we extend these undecidability results to show that it is possible to compute quantum mechanics for a local machine using a global, classical Turing machine with the same exponential quantum speed-up relative to the local machine's clock.

\subsection{Motivation and Related Works} 
The motivations of this work are both philosophical and practical. Philosophically, the simulation hypothesis has captivated pop-culture and science fiction for decades. However, the depictions in movies and books are not typically mathematically rigorous. Here, we provide a mathematically rigorous framework to build the simulation hypothesis using Turing machines\cite{Aaronson2013WhyComplexity}. These results apply to both quantum complexity theory and artificial intelligence. Specifically for quantum complexity theory, we construct a new oracle equivalence between quantum and classical computation using relativity, a topic that's been explored with closed timelike curves in general relativity \cite{Aaronson2009ClosedEquivalent,Deutsch1991QuantumLines}.
On the more practical side, the emergence of efficient semiconductor hardware for machine learning has ushered in a renaissance of A.I. New A.I. algorithms can construct fake, realistic human faces\cite{Karras2021ANetworks}, generate art from descriptions\cite{Ramesh2021Zero-ShotGeneration}, and even drive cars\cite{Bojarski2016EndCars}. Recent developments in brain-inspired computing may bring about the first hint of consciousness within A.I. \cite{Mehonic2022Brain-inspiredPlan, Bostrom2003AreSimulation, Rhodes2020Brain-inspiredCompleteness}. Once the A.I. reaches maturity, it may begin probing its own existence and try to determine whether it is just a machine, in the same way that we probe whether we are just an amalgamation of cells. Our work shows that without querying the answer from the machine simulating it, an A.I. cannot compute a set of properties we denote as \textit{simulation properties} that encompass several facts about how the artificial intelligence is being simulated. Applying these results to our own consciousness and universe, there are properties of the simulation hypothesis that we cannot compute. This leads to the following philosophical conjecture that proving we exist in a simulation is undecidable.

\begin{conjecture}
The problem of proving our universe is a simulation is undecidable.
\end{conjecture}

\section{Results}\label{sec2}
First, we formalize the simulation of any A.I. using the relative model shown in Figure \ref{fig:relative_model}.
\subsection{Relative Model} \label{sec:relative_model}
\begin{figure}[h]
    \centering
    \includegraphics[width=\textwidth]{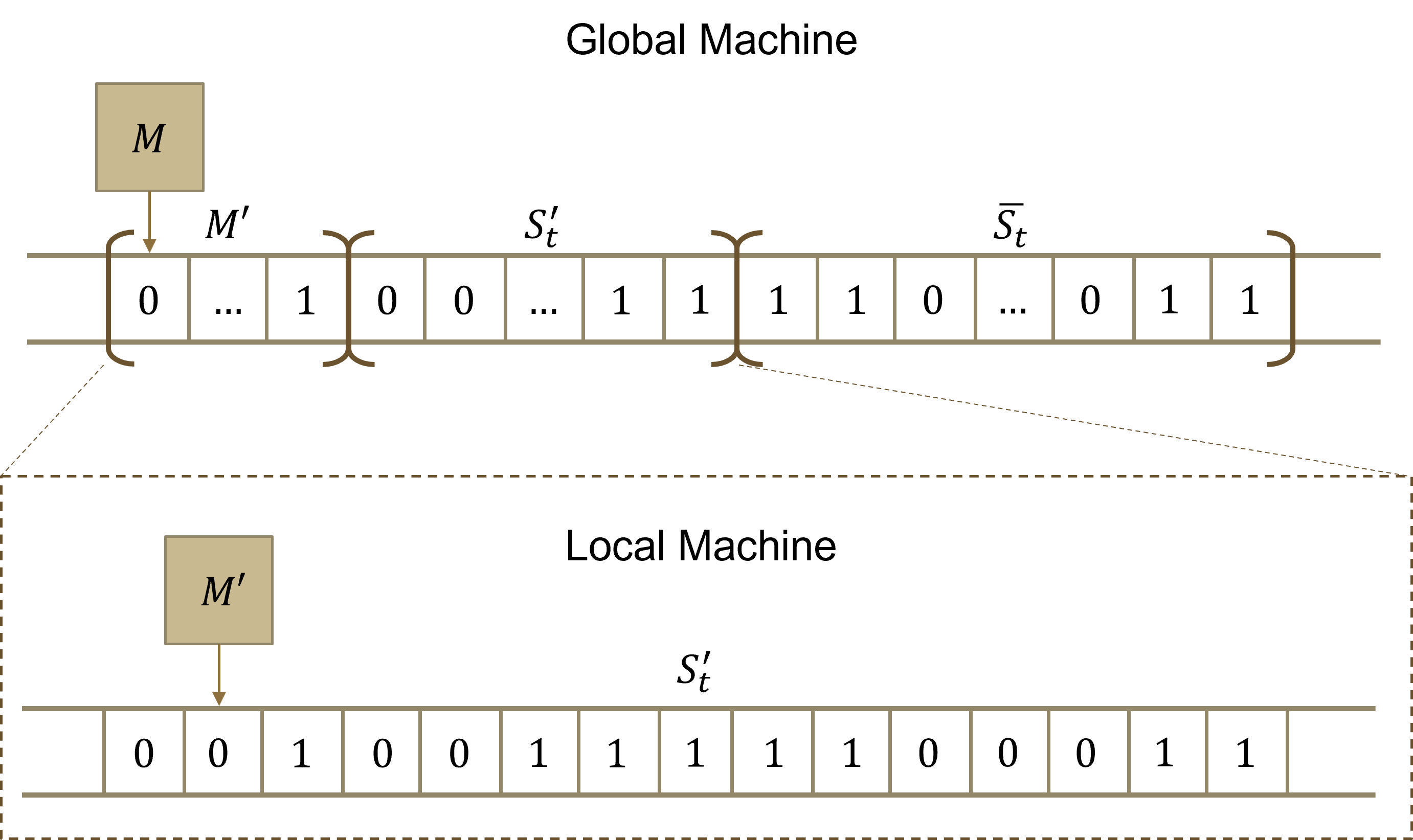}
    \caption{Relative model at some global time $t$. Here the global machine $M$ has its head position at the first bit that encodes the local machine $M'$. To the right of the encoding of $M'$ is a section of the tape $S'_t$ which acts as the local tape for $M'$ and is the only section of the tape that $M'$ can act on. Any computation that $M'$ would perform if it was real exists on $S'_t$. Finally, to the far right is the scrap paper portion of the tape $\bar{S}_t$ that $M$ uses to perform any necessary computations to simulate $M'$. Here, $\bar{S}_t$ can never be accessed by $M'$, i.e., $\bar{S}_t$ and $S'_t$ do not overlap}
    \label{fig:relative_model}
\end{figure}
The relative model of computation is another formalization of the universal Turing machine model used by Turing \cite{Turing1937OnEntscheidungsproblem}. We construct a global Turing machine $M = TM(\delta, Q, \Gamma)$ with a \textbf{global clock} $t \in \mathbb{Z}^+$ that operates according to its transition function $\delta$ on a string $S_t \in \Gamma^*$ encoded on its tape. Here, $Q$ is the set of states for $M$ and $\Gamma$ is a set of symbols that can be written on the tape, known as the global machine's alphabet. We refer to $S_t \in \Gamma^*$ as both the state of the tape at global time $t$ and the tape itself. Additionally, the global Turing machine $M$ simulates another machine $M'$ encoded in a portion of the tape $M' \in S_t$. The local machine has its own transition function $\delta'$, state set $Q'$, and alphabet $\Gamma'$ encoded in $M'$, though it is not necessarily isomorphic to a Turing machine. The tape of $M'$ exists as a subset of the overall tape $S'_{t} \subset S_{t}$. As the global Turing machine $M$ is performing the necessary computations to simulate $M'$, its transition function $\delta$ uses space on a separate portion of the tape $\bar{S}_t \subset S_t$ outside of $M'$ or $S'_t$. Once the global machine has performed all the necessary computations, it updates the local machine $M'$ and its tape $S'_t$ in accordance with its transition function $\delta'$ by copying over the new states. Once both are updated, we say the \textbf{local clock} $\tau$ of the local machine has moved forward one local time step. We define the set $K = \{k_1, k_2, ... \}$ to be all global time steps $k_\tau \in \mathbb{Z}^+$ at which the global machine $M$ has finished updating the tape state $S'_{t}$ and machine state $q'_t$ to match the output of the transition function $\delta'$ acting on the previously recorded global time step, i.e., $\delta'(S'_{k_\tau}, q'_{k_\tau}) \rightarrow \{S'_{k_{\tau + 1}}, q'_{k_{\tau+1}}\}$. We choose to increment $\tau$ only when $S'_{k_{\tau + 1}}$ and $q'_{k_{\tau+1}}$ are updated according to the transition function and not when $S'_{t}$ and $q'_{t}$ are changed by write operations for reasons outlined in Appendices \ref{app:write_order} and \ref{app:tau_choice}\footnote{We recommend the reader reads the full article before going to Appendices \ref{app:write_order} and \ref{app:tau_choice} as it relies on future sections.}. 

For example, consider Figure \ref{fig:relative_model}. If the global machine $M$ updates the local machine $M'$ at global time steps $t = 1$ and $t = 10$, then $k_1 = 1$ and $k_2 = 10$. At global time step $t = k_2$, the local machine $M'$ has only moved forward one local time step, i.e., from $\tau = 1$ to $\tau = 2$. We define this overall model of computation as a \textbf{relative model}.

%Then, for each global time step $k_\tau \in K$ we denote the number of bits required by the global machine $M$ to update the local machine's state $q'_t$ and tape $S'_t$ between global time steps $k_{\tau-1}$ and $k_{\tau}$ as $g_\tau$ and construct the set of space usage $G = \{g_1, g_2,...\}$ accordingly.

\begin{definition}[Relative Model $\mathcal{P}(M,M')$]
Consider a global Turing machine $M = TM(\delta, Q, \Gamma)$ that computes the tape state $S_t$ from an initial string $S_0 \in \Gamma^*$. There exists an encoding $M' \subset S_0$ for a local machine $M'$ such that the global machine $M$ simulates $M'$ for time $t \in \mathbb{Z}^+$ acting on a subset $S'_t$ of the tape $S_t$.
\end{definition}

Given the relative model, our objective is to construct a local machine $M'$ that can prove properties about its own simulation. We formalize these properties as \textit{simulation properties} and then immediately show that computing simulation properties is undecidable without querying $M$.
\subsection{Simulation Properties}
\begin{definition}[Simulation Property $R$]
 We define a simulation property $R : \{0,1\}^* \nrightarrow \{0,1\}^*$ as any partial, non-trivial function computed over the runtime tape set $\Delta_{\tau} = \{S_{k_{\tau}+1},...,S_{k_{\tau+1}-1}\}$.
\end{definition}
Here, a non-trivial function $f$ is a function whose image is not singleton, i.e., $\exists x,y: f(x) \ne f(y)$, and a partial function is a function $f : X \nrightarrow \{0,1\}^*$ defined over a subset of the domain, i.e., $X \subseteq \{0,1\}^*$\cite{Manin2010AMathematicians}. The runtime tape set is the set of all global tape states from global time step $k_{\tau}$ to global time step $k_{\tau + 1}$. Naturally, if the local machine cannot read these states, then it cannot compute the image of any function that depends on them.
\begin{problem}[SIMPROPERTY]
Consider a relative model $\mathcal{P}(M,M')$. Construct a local machine $M'$ to compute any simulation property $R$ without querying the solution from $M$.
\end{problem}
\begin{theorem} \label{thm:SIMPROPERTY_undecidable}
SIMPROPERTY is undecidable.
\end{theorem}
\begin{proof}
We will prove this by contradiction. If $R$ is non-trivial, then the image of $R$ depends on $\Delta_{\tau}$. So, to compute $R$, $M'$ needs access to $\Delta_{\tau}$. However, by definition of the relative clock $\tau$, $M'$ cannot perform any read or write operations on any tape states from $S_{k_{\tau}+1}$ to $S_{k_{\tau+1}}$ and so it cannot operate on any subset of $\Delta_{\tau}$ to compute $R$.
\end{proof}
 %For any probabilistic $\delta$ and $\delta'$, assuming a finite number of branches, 1. $M'$ still cannot observe $R$ directly and 2. $\delta'$ cannot encode any information that calculates $R$, for similar reasons. Because the branches are finite, $\delta'$ cannot just encode all possible values for $R$, as the range of values for any property $R$ of $M$ is not finite. Time and space required to compute one step of $M'$ can still be considered properties of $M$ because the only thing that changed was adding multiple possibilities to each step of computation. The probabilities that define all possible paths are still inherent to $M$ and thus time and space are as well.  Intuitively, this makes sense, as adding more paths for some $M$ to go down should make it harder for $M'$ to determine $R$. 
We now apply Theorem \ref{thm:SIMPROPERTY_undecidable} by showing that computing global time $t$ and global space in $\|S_t\|$ are simulation properties. Suppose the global machine $M$ requires $k_{\tau + 1} - k_\tau$ global time steps and $g_\tau$ bits in $S_t$ to compute the next state of $S'_{k_{\tau+1}}$. We formalize the problem of computing $t$ and $g_\tau$ as the SIMTIME and SIMSPACE problems, respectively, and show that computing them is undecidable in the same sense as the Halting problem by showing they require computing simulation properties.
\begin{problem}[SIMTIME]
Consider a relative model $\mathcal{P}(M,M')$. We denote the set of all global time steps where the global machine $M$ updates the local machine $M'$ as $K = \{k_1, k_2, ... \}$. Then, construct a local machine $M'$ to compute $k_{\tau+1} - k_\tau$ for any $k_\tau \in K$ without querying the solution from $M$.
\end{problem}
\begin{corollary} \label{cor:SIMTIME_undecidable}
Note that $k_{\tau+1} - k_\tau = \|\Delta_\tau\| + 1$ where $\|\Delta_\tau\|$ is the number of tape states in the runtime set $\Delta_\tau$. Then, $k_{\tau+1} - k_\tau$ is a partial function and dependent on the number of tape states in $\Delta_\tau$, i.e., it is non-trivial. Therefore, $k_{\tau+1} - k_\tau$ is a simulation property and thus SIMTIME is undecidable.
\end{corollary}
Notice that SIMTIME can be easily extended for probabilistic $M$ and $M'$ and we do so in Appendix \ref{app:pSIMTIME}.
\begin{problem}[SIMSPACE]
Consider a relative model $\mathcal{P}(M,M')$. We denote the set of all global time steps where the global machine $M$ updates the local machine $M'$ as $K = \{k_1, k_2, ... \}$. Then, for each global time step $k_\tau \in K$, we denote the number of bits required by the global machine $M$ to compute the local machine $M'$ between global times $k_{\tau-1}$ to $k_{\tau}$ as $g_\tau$. Then, construct a local machine $M'$ to compute $g_{\tau}$ for any $\tau \in \mathbb{Z}^+$ without querying the solution from $M$.
\end{problem}
\begin{corollary} \label{cor:SIMSPACE_undecidable}
Computing $g_\tau$ requires counting the number of bits used by $M$ in the runtime set $\Delta_\tau$, which is a partial, non-trivial function of $\Delta_\tau$ and therefore a simulation property. Thus, SIMSPACE is undecidable.
\end{corollary}
Both undecidability results for SIMTIME and SIMSPACE can give rise to special relativistic effects when we allow the local machine to query functions from the global machine which are not efficiently computable with respect to the global clock $t$.
\subsubsection{Relative Oracles}
Suppose the local machine $M'$ queries the global machine $M$ to compute a function $f(x)$ at global time $k_{\tau}$. To do so, we create a query state $q'_{f} \in Q'$ such that when the local machine $M'$ enters state $q'_f$, the global machine $M$ computes a function $f(x)$ that operates on a substring $x \subseteq S'_{k_\tau}$ and the global machine writes the output on a portion of $S'_{k_{\tau+1}}$. By Corollaries \ref{cor:SIMTIME_undecidable} and \ref{cor:SIMSPACE_undecidable}, the local machine will not be able to compute the global time or space required by the global machine $M$ to compute $f(x)$. So, we can easily see how from the perspective of the local machine, the function is computed in only one local time step and takes up only the necessary number of bits to write the output. To consider the asymptotic nature of how the global resources required by $M$ increase with respect to the local resources of the local machine, we denote the different resource complexities using the subscript of the clock or space variable under the asymptotic symbol. For example, we say $f(x) = O_\tau(g(n))$, where $x \in \{0,1\}^n$, is equivalent to saying the function $f(x)$ can be computed in $O(g(n))$ local time steps with respect to its local clock $\tau$ under a relative model. Now, we argue that for any queried function $f(x)$, the local machine $M'$ can have better asymptotic computational power relative to its local time $\tau$ and space $\|S_t'\|$, where $\|S_t'\|$ denotes the number of non-blank symbols on $S'_t$, than the global machine $M$ can achieve with respect to the global time $t$ and space $\|S_t\|$.

%Additionally, we can say $f(x) = O_{\|S_t'\|}(Y))$ is equivalent to saying the function $f(x)$ can be computed in $O(Y)$ "local bits" (i.e. just the number of bits required to encode the output of $f$) while using $\|S_t'\|$ extra bits outside of $M'$.

\begin{problem}[Relative Oracle Problem]
Given any computable function $f \in \mathcal{F}^*$, construct a relative model $\mathcal{P}(M,M')$ that computes $f$ in $O_\tau(1)$ time while only using $O_{\|S_t'\|}(\|f(x)\|))$ space where $\|f(x)\|$ is the number of bits required to encode the output of $f$. 
\end{problem}
\begin{theorem}
There exists a relative model $\mathcal{P}(M,M')$ that solves the Relative Oracle problem. \label{thm:seminary_problem}
\end{theorem}
\begin{proof}
Because $f$ is computable, construct a query state for the local machine $M'$ that queries the global machine $M$ to compute $f(x)$ in $\bar{S}_t$. Given the output of $f(x)$, the global machine $M$ requires $\|f(x)\|$ write operations to update $S'_{t}$. Note that the local clock won't increment until the transition function $\delta'$ is computed, which requires writing all of the output of $f(x)$ into $S'_t$ first. Therefore, $\tau$ will only increment after $M$ writes all of $f(x)$ into $S'_t$, hence $O_\tau(1)$ time. Likewise, the encoding of the output of $f$ on $S'_{t}$ will always require $\|f(x)\|$ space on the tape, hence the space complexity is always $\Theta_{\|S_t'\|}(\|f(x)\|))$. 
\end{proof}
%\begin{proof}
%Because there exists sufficient space on the tape outside of the encoding of $S'_{t}$, construct an $M$ that computes the desired function outside of $S'_t$ and then update $S'_t$ only after $M$ has computed the desired function. 
%\end{proof}
 By simply knowing the minimal global complexities for a given function, Theorem \ref{thm:seminary_problem} implies that we can provably construct a local machine that has better local computational powers than the global machine.

\begin{corollary}
For any computable function $f \in \mathcal{F}^*$ that requires $\Omega_t(1)$ global time or $\Omega_{\|S\|}(\|f(x)\|)$ global space to compute but only $\|f(x)\|$ bits to encode the output of $f$, Theorem \ref{thm:seminary_problem} implies there exists a relative model $\mathcal{P}(M,M')$ that can compute $f$ in $\Theta_\tau(1)$ time and $\Theta_{\|S_t'\|}(\|f(x)\|)$ space. \label{cor:joan_wins}
\end{corollary}

Corollary \ref{cor:joan_wins} implies that given a computable function whose global time complexity is greater than constant time or the global space complexity is greater than the minimal write space, the local machine can locally be more powerful than the global machine. Corollaries \ref{cor:SIMTIME_undecidable} and \ref{cor:SIMSPACE_undecidable} show that the local machine can't have any information about the global passage of time or utilized space for computing any function $f(x)$ and so it gives rise to special relativistic effects like time dilation and length contraction which are quantified with Lorentz factors \cite{Griffiths2005Electrodynamics}. 

\subsubsection{Lorentz Factors of Relative Models}\label{sec:special_relativity}
The global machine $M$ is observing a global clock $t$ and global space $S_t$ while the local machine $M'$ is observing a local clock $\tau$ and local space $S'_t$. From the global machine's frame, time is moving at a steady pace and all computations are performed at the rate of the global clock $t$. However, Corollaries \ref{cor:SIMTIME_undecidable} and \ref{cor:SIMSPACE_undecidable} prove that the only time and space that the local machine can even compute are its local time $\tau$ and local space $S'_t$. In fact, from the perspective of the local machine, all write operations by the global machine to the local tape $S'_t$ appear to happen simultaneously (Appendix \ref{app:write_order}). Therefore, the local machine's frame only perceives the local time $\tau$ and local space $S'_t$. These effects are equivalent to the time dilation and length contraction effects seen in special relativity \cite{Griffiths2005Electrodynamics}. In special relativity, time dilation and length contraction effects occur when we compute the time and length intervals of events from two observers. To experience these effects, one observer must have a relative velocity close to the speed of light with respect to the other observer. Then, time dilation is when the time experienced by one observer is faster than the time observed by another observer. Likewise, length contraction is the effect where each observer experiences a different length scale for the same object. In a relative model, time dilation occurs between the global and local times of the models and length contraction occurs for the required number of bits to perform a computation. Unlike in special relativity where the Lorentz factor is determined by the relative velocity between observers, in our model the queried function determines the Lorentz factors for time and space resources. We compute a Lorentz factor for time dilation between the global time $t$ and local time $\tau$ via the ratio between global and local time intervals, i.e.,
\begin{equation}
    \gamma_{k_\tau} = \frac{k_{\tau+1} - k_\tau}{\tau + 1 - \tau} = k_{\tau+1} - k_\tau. \label{eq:lorentz_time}
\end{equation}
However, the Lorentz factor for length contraction is computed using the global space ratio to be
\begin{equation}
    \gamma_{g_\tau} = \frac{g_{\tau+1} - g_\tau}{\|f(x)\|},\label{eq:lorentz_space}
\end{equation}
where $f(x)$ is the output of any computation performed in $\bar{S}_t$ and copied over to $S'_t$. In physics, the Lorentz factor for time and space are equal but for relative models the Lorentz factors for time and space depend on the particular function $f(x)$ and global machine $M$ \footnote{Furthering the analogy, one can imagine that if $\gamma_{k_\tau}$ is constant for an interval of local times $\tau$, then the computation time is proceeding at a constant velocity. However, if $\gamma_{k_\tau}$ varies between local time intervals, then the computation is undergoing an acceleration. Hence, there may be a connection to general relativity through accelerating frames of relative models.}. Due to this dependence, we can immediately see that the Lorentz factors are also simulation properties.
\begin{corollary}
Because the Lorentz factors (\ref{eq:lorentz_time}) and (\ref{eq:lorentz_space}) are dependent on computing SIMTIME and SIMSPACE, which are both simulation properties, the Lorentz factors are also simulation properties and therefore undecidable by Theorem \ref{thm:SIMPROPERTY_undecidable}.
\end{corollary}

 Using these special relativistic effects, the local machine is effectively a local oracle model that only perceives its own time and space \cite{Arora2009ComputationalApproach}. We include a discussion on comparing relative models and oracle models in Appendix \ref{app:oracles}. This local oracle model allows us to construct a fully classical global model that locally preserves the exponential speed-ups of quantum mechanics. However, before we can construct this local quantum model, we have to consider a necessary component to simulating quantum mechanics, error. 

 \subsection{Undecidability of Computing Global Machine Error}
 Suppose that the global machine $M$ wants to save global time and space resources by approximating a queried function $f(x)$ instead of fully computing it. The function that the local machine uses to compute the error incurred by the global machine's approximation is no longer a simulation property so Theorem \ref{thm:SIMPROPERTY_undecidable} is not applicable. Clearly, if the machines can perfectly encode $f(x)$ and the global machine returns an approximation $\Tilde{f}(x)$ to $f(x)$, and the local and global machines both have access to $x$, then the local machine can simply compute $f(x)$ and catch the global machine in an error. However, if the local machine only has approximations to $x$ in the form of a set $\Tilde{X}$ of possible $x$ (for example, some range in $\mathbb{Z}$), then the global machine can approximate $f(x)$ be leveraging the inherent error in $\Tilde{X}$ and returning a string from the set of possible strings $\{f(\Tilde{x}) : \Tilde{x} \in \Tilde{X}\}$. This kind of scenario arises when performing measurements. Suppose there exists a section of the tape labeled $m_t \subset S'_t$ that contains a string $x$ at time $t$. Suppose that due to the encoding of $M'$, it is impossible to perfectly measure $x$. In other words, whenever $M'$ tries to encode $x$ in a separate portion of the tape, $x$ is approximated to be $x'$ by sampling from a measurement function $x' \sim g(x)$. However, $x'$ is permanently stored in the local machine's tape. So, it can reference $x'$ whenever it needs. Given a measured $x'$, there exists a set $\Tilde{X}$ of possible $x$ that could have generated $x'$. Now, at some time $t_2$, $M$ will compute $f(x)$ and overwrite $m_t$. $M$ wants to exploit the uncertainty of $\Tilde{X}$ from measurements to approximate $f(x)$ using a function $\Tilde{f}(x)$ without being detected by $M'$. We formalize the problem of determining if the global machine returned $\Tilde{f}(x)$ as follows.
 
\begin{problem}[MEASURE]
Consider a relative model $\mathcal{P}(M,M')$. Let $\Tilde{f}(x)$ be an approximation of a computable function $f(x) \in \mathcal{F}^*$. Suppose there exists some string $x \subset S'_t$ which $M'$ measures with a function $g(x)$ that returns a value $x' \sim g(x)$. Define $\Tilde{X}$ to be the set of all possible $x$ given a measurement $x' \sim g$. Suppose the local machine $M'$ queries the function $f(x)$ at time $k_\tau$ and the global machine $M$ randomly chooses to compute $\Tilde{f}(x)$ or $f(x)$ and encode the solution in a substring $y \subset S'_{k_{\tau+1}}$. Then, construct a local machine $M'$ to determine if the global machine $M$ computed $\Tilde{f}(x)$ or $f(x)$.
\end{problem}
\begin{theorem}
MEASURE is undecidable when the approximation error of $\Tilde{f}(x)$ is within the possible values of $f(x)$ introduced by the approximation of $x$ by the set $\Tilde{X}$, i.e., $\Tilde{f}(x) \cup (\bigcup_{\Tilde{x} \in \Tilde{X}}f(\Tilde{x})) = \bigcup_{\Tilde{x}\in \Tilde{X}}f(\Tilde{x})$. \label{thm:undecidable_abs_error}
\end{theorem}
\begin{proof}
We will prove this by contradiction. Assume the local machine $M'$ can determine which function the global machine $M$ used. Because $M'$ doesn't know $x$, but instead it has a set $\Tilde{X}$ of possible $x$, the local machine $M'$ must compute the image of $f(x)$ for all possible values of $x$ in $\Tilde{X}$ and see if the output $y \in S'_t$ is in the image of $f$ over $\Tilde{X}$, i.e., $M'$ computes the set $\textbf{Image}(f)_{\Tilde{X}} = \{f(x_i) : \forall x_i \in \Tilde{X}\}$. If the returned output $y$ lies outside of $\textbf{Image}(f)_{\Tilde{X}}$, then the local machine $M'$ can conclude that the global machine $M$ used $\Tilde{f}$. However, if it lies inside of $\textbf{Image}(f)_{\Tilde{X}}$ such that $y = f(x_i)$ for some $x_i \in \Tilde{X}$, then for $M'$ to know that $y$ wasn't computed by $f(x)$, it must know that $x_i \ne x$, which contradicts the definition of $\Tilde{X}$. Therefore, if $\Tilde{f}(x) \in \bigcup_{x_i \in \Tilde{X}} f(x)$, then $M'$ cannot determine which function $M$ used.
\end{proof}

We can easily extend the context of Theorem \ref{thm:undecidable_abs_error} to show an even more powerful result. Namely, we show that the global machine $M$ can always trick the local machine $M'$ to accept in $O_\tau(1)$ time so long as there exists at least one state in which $M'$ will accept and it is guaranteed to be computable in a finite amount of steps. This means no matter what the error introduced by an approximation, the global machine can always force the local machine into a known state. However, this comes at the cost of global time.
  
\begin{theorem}
  Consider a relative model $\mathcal{P}(M,M')$. Suppose it is guaranteed that $M'$ will enter the ACCEPT or REJECT states within some finite local time $\mathcal{T}$ after the current local time $\tau$, and there exists at least one state $S'_{\text{ACCEPT}}$ such that $M'$ will accept given the current setting of the tape $S'_{k_\tau}$ and state $q'_{k_\tau}$. Then, $M$ can always prepare $S'_{k_{\tau+1}} = S'_{\text{ACCEPT}}$ to guarantee a later accept in $O_\tau(1)$ time. \label{thm:spoof_ACCEPT}
 \end{theorem}
\begin{proof}
  Simply have the global machine $M$ simulate $M'$ up until local time $\mathcal{T} + \tau$ on a section of the tape $\bar{S}_t$ outside of $S'_t$ with different tape states until $M$ finds one that will force $M'$ to accept by local time $\tau + \mathcal{T}$. Because there exists at least one initial state of the tape $S'_{\text{ACCEPT}}$ where $M'$ will accept, $M$ is guaranteed to eventually find it. Once it does so, prepare the next state of the tape $S'_{k_{\tau+1}}$ to be $S'_{\text{ACCEPT}}$ for the next local time step.
\end{proof}

We can see that due to Theorems \ref{thm:undecidable_abs_error} and \ref{thm:spoof_ACCEPT}, the global machine $M$ can force $M'$ to not perceive the error either by returning a value within the allowed error that $M'$ will accept, i.e., Theorem \ref{thm:undecidable_abs_error}, or it can force it to accept that no error has occurred, i.e., Theorem \ref{thm:spoof_ACCEPT}. Now we have everything in place to simulate quantum mechanics with the relative model.

\subsection{A Relative Model for Quantum State Evolution}
  Following the Church-Turing-Deutsch Principle, if we can compute quantum mechanics, then we can simulate any process in the universe to an arbitrary accuracy (assuming that the laws of quantum physics can completely describe every physical process)\cite{Deutsch1985QUANTUMCOMPUTER.,Lloyd1996UniversalSimulators}. So, we show how to simulate quantum mechanics using a relative model $\mathcal{P}(M,M')$ with the same local efficiency that we experience in our universe. Let $M$ be a global Turing machine and let $M'$ encode the Schr\"{o}dinger equation acting on a set of qubits. Because the global machine $M$ has access to the entire tape, it can compute the Hamiltonian $\hat{H}(\tau)$ for any subset of particles in the universe at any local time $\tau$. However, to satisfy the third law of thermodynamics and avoid zero entropy, every particle must be coupled to at least one other particle via intermediate interactions\cite{Deutsch1985QUANTUMCOMPUTER.}. So, to accurately describe the time evolution of the universe, $M$ must construct a Hamiltonian $\hat{H}$ and quantum state vector $\vert \psi(\tau) \rangle$ for the entire universe at every local time step $\tau$. By Theorem \ref{thm:seminary_problem}, this can be done in $O_\tau(1)$ time between every local time step of the universe. Luckily, the universe is a closed quantum system with a time-independent Hamiltonian $\hat{H}$ \cite{Sakurai2017ModernNapolitano.}. We don't need a measurement postulate in this context because we are describing the entire universe as a single state vector and assume that a measurement collapse can be encoded as a byproduct of the quantum state evolution \cite{Masanes2019TheRedundant}. Consider a quantum state vector
  \begin{align}
      \vert \psi \rangle = c_1 \vert \psi_1 \rangle + ... + c_{2^n} \vert \psi_{2^n} \rangle.
  \end{align}
    We can encode any quantum state $\vert \psi \rangle$ in $S'_t$ and update it at every local clock time in $O_\tau(1)$ time. So, we introduce a local time-dependence $\tau$ on the state $\vert \psi(\tau) \rangle$ and compute it using the Schr\"{o}dinger equation 
  \begin{equation}
      i \hbar \frac{d}{d\tau} \vert \psi(\tau) \rangle = \hat{H} \vert \psi(\tau) \rangle. \label{eq:schrodinger}
  \end{equation}
   Using the time-independent Hamiltonian $\hat{H}$, we can solve (\ref{eq:schrodinger}) using the time-independent Schr\"{o}dinger equation 
  \begin{align}
      \vert \psi(\tau) \rangle = e^{-\frac{i}{\hbar} \hat{H} \tau}\vert \psi(0) \rangle \label{eq:time_ind_schrod}
  \end{align}
  To compute (\ref{eq:time_ind_schrod}) to an arbitrary accuracy, we expand the exponential explicitly as a sum of matrix products
  \begin{align}
      e^{-\frac{i}{\hbar} \hat{H} \tau} = \sum_{j = 0}^{\infty}\frac{(-i \tau/\hbar)^j}{j!}\hat{H}^j. \label{eq:exp_expanded}
  \end{align}
  Immediately we see an issue. If the infinite sum (\ref{eq:exp_expanded}) cannot be simplified, then the summation is required to go to infinity. So, the global machine $M$ will never finish computing (\ref{eq:exp_expanded}). However, we can approximate (\ref{eq:exp_expanded}) up to some maximum term and return an approximation $\vert \Tilde{\psi}(\tau) \rangle$\footnote{An added bonus of the time-independent Schrodinger equation is that the error does not accumulate between local time steps $\tau$ because one can compute the state $\vert \psi(\tau) \rangle$ using only $\hat{H}$, $\tau$, and $\vert \psi(0) \rangle$.}. By Theorem \ref{thm:undecidable_abs_error}, the global machine $M$ only needs to compute the approximation $\vert \Tilde{\psi}(\tau) \rangle$ using (\ref{eq:exp_expanded}) within an allowable error that makes the problem of detecting the error by the local machine $M'$ undecidable. Likewise, by Corollaries \ref{cor:SIMTIME_undecidable} and \ref{cor:SIMSPACE_undecidable}, the local machine $M'$ will not perceive the exponentially large amount of time or space required for $M$ to compute (\ref{eq:exp_expanded}) within this error.   
  \begin{corollary}
  There exists a relative model $\mathcal{P}(M,M')$ that approximates the Schr\"{o}dinger equation $f(\vert\psi(\tau)\rangle)$ within an arbitrary error in $O_\tau(1)$ time and $\Theta_{\|S_t'\|}(\|f(\vert\psi\rangle)\|)$ space where $\|f(\vert\psi\rangle)\|$ is the desired encoding precision. \label{cor:schrodinger}
  \end{corollary}

%\subsection{Testing Randomness in the Base Machine}
%If the base machine is using a pseudo-boolean generator

%\subsubsection{Simulating Measurement Collapse}
%While there are several interesting theories for the collapse of the eigenstate of quantum measurements, we propose a simplified, equivalent 

%\subsection{Probing Discretization}
%If the Relative Church-Turing-Deutsch Thesis is true, then there must exist some discretization to $\tau$ that is detectable. However, we show that even computing this discretization is also undecidable.

% \begin{problem}[SIMSPACE]
% Consider a tape that can encode any string $S \in \{0,1\}^*$. Let there exist a Turing machine $M$ that acts on the tape by some probabilistic transition function $\delta$ under a distribution $p$. Suppose there exists an encoding $S_{s,t'}$ for a machine $M_s$ at time $t'$ with a probabilistic transition function $\delta_s$ under a distribution $q$ such that $M$ simulates $M_s$ for time $t \geq t'$. We denote the set of all global distance steps 

% as $K = \{k_1, k_2, ... \}$. Then, construct a machine encoding $S_{s,t'}$ that computes $k_{\tau+1} - k_\tau$ for any $k_\tau \in K$.
% \end{problem}
% \begin{theorem}
% SIMSPACE is undecidable.
% \end{theorem}
% \begin{proof}
% \end{proof}

% Assuming that quantum mechanics is an accurate theory, then our universe is quantum-Turing equivalent. However, we also know it's possible to simulate quantum mechanics with exponential overhead. 
% \begin{proposition}

% \end{proposition}

\section{Discussion}\label{sec12}

\subsection{Relative Church-Turing-Deutsch Thesis}
Theorem \ref{thm:seminary_problem} gives us a method of relativizing the physical Church-Turing thesis \cite{Piccinini2007ComputationalismFallacy}. Consider the scenario where our universe is constructed as the local machine $M'$ in a relative model. There is a global machine $M$ computing the physical laws encoded in $M'$. By Corollaries \ref{cor:SIMTIME_undecidable} and \ref{cor:SIMSPACE_undecidable} as discussed in \ref{sec:special_relativity}, we experience the local time $\tau$ and local space $S'_t$ to our universe and cannot perceive the global time $t$ or space $S$ taken by the global machine. By Corollary \ref{cor:schrodinger}, we know it is possible for quantum mechanics to be computed in $\Theta_\tau(1)$ time as we experience in our universe. Therefore, it is entirely possible that a physical Church-Turing thesis is globally and physically true for all computational models in our universe; meaning Turing machines are the only computational models that obey the global mathematics of our universe and our information theoretic intuitions about the process of computation. However, locally we have access to a physical quantum Church-Turing thesis through the queried functions that govern quantum mechanics, like the Schr\"{o}dinger equation\cite{Kaye2006AnComputing}. The special relativistic effects from Corollaries \ref{cor:SIMTIME_undecidable} and \ref{cor:SIMSPACE_undecidable} guarantee that it's impossible for us to perceive the exponential overhead in time and space required to compute quantum mechanics with the global Turing machine.

%A natural corollary of this is that as quantum computers have access to query $M$ via quantum mechanics, they have an inherent and guaranteed speedup in comparison to classical computers.  Quantum computers can take a step in line with the update of the universe, evolving with the schr\"odinger equation in \Theta(1) time, whereas any classical TM would have to perform all the computations that M does in order to perform the same computation. As the computation of the quantum mechanics of the universe requires the calculation of a infinite sum, terminating after an amount of time which is sufficient for the relative representational error of $M'$ takes exponentially longer than on a quantum TM [this needs to be more rigorous].  Thus quantum TMs are more powerful than classical ones, and we have disproved the complexity-theoretic church-turing thesis :).

Unfortunately, our measurement devices always have some guaranteed error due to encoding. So long as the error introduced by approximating the Schr\"{o}dinger equation (\ref{eq:exp_expanded}) is within the error of all of our measurement instruments, Theorem \ref{thm:undecidable_abs_error} protects us from detecting the error in the universe's computation. Not to mention, we don't know a priori the function that the universe is computing. Theorem \ref{thm:undecidable_abs_error} relies on an agreed function $f(x)$ but we've only deduced the Schr\"{o}dinger equation from observations. The physicist's assumption that whatever the universe provides is inherently a law of the universe means that we will accept any approximation from the universe as \textbf{the} law. So, all errors from the universe are undetectable.  Putting all of this together, we have a more comprehensive simulation hypothesis of the universe that preserves any physical laws as transition rules and queried functions.

\section{Conclusion}\label{sec13}
In this work, we study two questions at the heart of the simulation hypothesis, i) is it possible for a simulation to decide it is a simulation and ii) is quantum mechanics compatible with a classical simulation hypothesis? To address these questions, we construct a relative model of computation involving a global Turing machine $M$ simulating a local machine $M'$. Global machine $M$ bestows more computational power to local machine $M'$ relative to the local machine's space and time resources by allowing $M'$ to query $M$ to compute any function for it on its next update. We show that the problem of the local machine $M'$ computing the simulation time, space, or error is undecidable in the same sense as the Halting problem. We further prove a connection to the time dilation and length contraction effects seen in special relativity and we propose a relative Church-Turing-Deutsch thesis where a classical, global simulator constructs a local quantum mechanical model for the observable universe. We show that this model preserves both the classical Church-Turing thesis from the global machine's global frame and the quantum Church-Turing thesis from the frame of our observable universe. Therefore, our model preserves the exponential speed-up given by quantum computers and provides a reconciliation between the intuitive notion of Turing machines and the unintuitive linear time evolution of an exponentially large quantum state in Hilbert space.

\backmatter

\bmhead{Acknowledgments}

We'd like to acknowledge the following people for their input on this work and essential discussions: Vahagn Mhkitaryan, Sam Peana and Colton Fruhling.

%We'd like to acknowledge Scott Aaronson for 

\bibliography{references}


\begin{thebibliography}{20}
\ifx \bisbn   \undefined \def \bisbn  #1{ISBN #1}\fi
\ifx \binits  \undefined \def \binits#1{#1}\fi
\ifx \bauthor  \undefined \def \bauthor#1{#1}\fi
\ifx \batitle  \undefined \def \batitle#1{#1}\fi
\ifx \bjtitle  \undefined \def \bjtitle#1{#1}\fi
\ifx \bvolume  \undefined \def \bvolume#1{\textbf{#1}}\fi
\ifx \byear  \undefined \def \byear#1{#1}\fi
\ifx \bissue  \undefined \def \bissue#1{#1}\fi
\ifx \bfpage  \undefined \def \bfpage#1{#1}\fi
\ifx \blpage  \undefined \def \blpage #1{#1}\fi
\ifx \burl  \undefined \def \burl#1{\textsf{#1}}\fi
\ifx \doiurl  \undefined \def \doiurl#1{\url{https://doi.org/#1}}\fi
\ifx \betal  \undefined \def \betal{\textit{et al.}}\fi
\ifx \binstitute  \undefined \def \binstitute#1{#1}\fi
\ifx \binstitutionaled  \undefined \def \binstitutionaled#1{#1}\fi
\ifx \bctitle  \undefined \def \bctitle#1{#1}\fi
\ifx \beditor  \undefined \def \beditor#1{#1}\fi
\ifx \bpublisher  \undefined \def \bpublisher#1{#1}\fi
\ifx \bbtitle  \undefined \def \bbtitle#1{#1}\fi
\ifx \bedition  \undefined \def \bedition#1{#1}\fi
\ifx \bseriesno  \undefined \def \bseriesno#1{#1}\fi
\ifx \blocation  \undefined \def \blocation#1{#1}\fi
\ifx \bsertitle  \undefined \def \bsertitle#1{#1}\fi
\ifx \bsnm \undefined \def \bsnm#1{#1}\fi
\ifx \bsuffix \undefined \def \bsuffix#1{#1}\fi
\ifx \bparticle \undefined \def \bparticle#1{#1}\fi
\ifx \barticle \undefined \def \barticle#1{#1}\fi
\bibcommenthead
\ifx \bconfdate \undefined \def \bconfdate #1{#1}\fi
\ifx \botherref \undefined \def \botherref #1{#1}\fi
\ifx \url \undefined \def \url#1{\textsf{#1}}\fi
\ifx \bchapter \undefined \def \bchapter#1{#1}\fi
\ifx \bbook \undefined \def \bbook#1{#1}\fi
\ifx \bcomment \undefined \def \bcomment#1{#1}\fi
\ifx \oauthor \undefined \def \oauthor#1{#1}\fi
\ifx \citeauthoryear \undefined \def \citeauthoryear#1{#1}\fi
\ifx \endbibitem  \undefined \def \endbibitem {}\fi
\ifx \bconflocation  \undefined \def \bconflocation#1{#1}\fi
\ifx \arxivurl  \undefined \def \arxivurl#1{\textsf{#1}}\fi
\csname PreBibitemsHook\endcsname

%%% 1
\bibitem{Turing1937OnEntscheidungsproblem}
\begin{botherref}
\oauthor{\bsnm{Turing}, \binits{A.M.}}:
{On computable numbers, with an application to the entscheidungsproblem}.
Proceedings of the London Mathematical Society
\textbf{s2-42}(1)
(1937).
\doiurl{10.1112/plms/s2-42.1.230}
\end{botherref}
\endbibitem

%%% 2
\bibitem{Arora2009ComputationalApproach}
\begin{bbook}
\bauthor{\bsnm{Arora}, \binits{S.}},
\bauthor{\bsnm{Barak}, \binits{B.}}:
\bbtitle{{Computational Complexity: A Modern Approach}},
\bedition{1st} edn.
\bpublisher{Cambridge University Press},
\blocation{USA}
(\byear{2009})
\end{bbook}
\endbibitem

%%% 3
\bibitem{Turing1950ComputingTuring}
\begin{botherref}
\oauthor{\bsnm{Turing}, \binits{A.M.}}:
{Computing machinery and intelligence-AM Turing}.
Mind
\textbf{59}(236)
(1950)
\end{botherref}
\endbibitem

%%% 4
\bibitem{Bostrom2003AreSimulation}
\begin{botherref}
\oauthor{\bsnm{Bostrom}, \binits{N.}}:
{Are you living in a computer simulation?}
Philosophical Quarterly
\textbf{53}(211)
(2003)
\end{botherref}
\endbibitem

%%% 5
\bibitem{Deutsch1985QUANTUMCOMPUTER.}
\begin{botherref}
\oauthor{\bsnm{Deutsch}, \binits{D.}}:
{QUANTUM THEORY, THE CHURCH-TURING PRINCIPLE AND THE UNIVERSAL QUANTUM
  COMPUTER.}
Proceedings of The Royal Society of London, Series A: Mathematical and Physical
  Sciences
\textbf{400}(1818)
(1985).
\doiurl{10.1098/rspa.1985.0070}
\end{botherref}
\endbibitem

%%% 6
\bibitem{Griffiths2005Electrodynamics}
\begin{botherref}
\oauthor{\bsnm{Griffiths}, \binits{D.J.}},
\oauthor{\bsnm{Inglefield}, \binits{C.}}:
{ Introduction to Electrodynamics }.
American Journal of Physics
\textbf{73}(6)
(2005).
\doiurl{10.1119/1.4766311}
\end{botherref}
\endbibitem

%%% 7
\bibitem{Aaronson2013WhyComplexity}
\begin{bchapter}
\bauthor{\bsnm{Aaronson}, \binits{S.}}:
\bctitle{{Why philosophers should care about computational complexity}}.
In: \bbtitle{Computability: Turing, Godel, Church, and Beyond},
(\byear{2013}).
\doiurl{10.7551/mitpress/8009.003.0011}
\end{bchapter}
\endbibitem

%%% 8
\bibitem{Aaronson2009ClosedEquivalent}
\begin{botherref}
\oauthor{\bsnm{Aaronson}, \binits{S.}},
\oauthor{\bsnm{Watrous}, \binits{J.}}:
{Closed timelike curves make quantum and classical computing equivalent}.
Proceedings of the Royal Society A: Mathematical, Physical and Engineering
  Sciences
\textbf{465}(2102)
(2009).
\doiurl{10.1098/rspa.2008.0350}
\end{botherref}
\endbibitem

%%% 9
\bibitem{Deutsch1991QuantumLines}
\begin{botherref}
\oauthor{\bsnm{Deutsch}, \binits{D.}}:
{Quantum mechanics near closed timelike lines}.
Physical Review D
\textbf{44}(10)
(1991).
\doiurl{10.1103/PhysRevD.44.3197}
\end{botherref}
\endbibitem

%%% 10
\bibitem{Karras2021ANetworks}
\begin{botherref}
\oauthor{\bsnm{Karras}, \binits{T.}},
\oauthor{\bsnm{Laine}, \binits{S.}},
\oauthor{\bsnm{Aila}, \binits{T.}}:
{A Style-Based Generator Architecture for Generative Adversarial Networks}.
IEEE Transactions on Pattern Analysis and Machine Intelligence
\textbf{43}(12)
(2021).
\doiurl{10.1109/TPAMI.2020.2970919}
\end{botherref}
\endbibitem

%%% 11
\bibitem{Ramesh2021Zero-ShotGeneration}
\begin{bchapter}
\bauthor{\bsnm{Ramesh}, \binits{A.}},
\bauthor{\bsnm{Pavlov}, \binits{M.}},
\bauthor{\bsnm{Goh}, \binits{G.}},
\bauthor{\bsnm{Gray}, \binits{S.}},
\bauthor{\bsnm{Voss}, \binits{C.}},
\bauthor{\bsnm{Radford}, \binits{A.}},
\bauthor{\bsnm{Chen}, \binits{M.}},
\bauthor{\bsnm{Sutskever}, \binits{I.}}:
\bctitle{{Zero-Shot Text-to-Image Generation}}.
In: \beditor{\bsnm{Meila}, \binits{M.}},
\beditor{\bsnm{Zhang}, \binits{T.}} (eds.)
\bbtitle{Proceedings of the 38th International Conference on Machine Learning}.
\bsertitle{Proceedings of Machine Learning Research},
vol. \bseriesno{139},
pp. \bfpage{8821}--\blpage{8831}.
\bpublisher{PMLR}, \blocation{???}
(\byear{2021}).
\burl{https://proceedings.mlr.press/v139/ramesh21a.html}
\end{bchapter}
\endbibitem

%%% 12
\bibitem{Bojarski2016EndCars}
\begin{botherref}
\oauthor{\bsnm{Bojarski}, \binits{M.}},
\oauthor{\bsnm{Del~Testa}, \binits{D.}},
\oauthor{\bsnm{Dworakowski}, \binits{D.}},
\oauthor{\bsnm{Firner}, \binits{B.}},
\oauthor{\bsnm{Flepp}, \binits{B.}},
\oauthor{\bsnm{Goyal}, \binits{P.}},
\oauthor{\bsnm{Jackel}, \binits{L.D.}},
\oauthor{\bsnm{Monfort}, \binits{M.}},
\oauthor{\bsnm{Muller}, \binits{U.}},
\oauthor{\bsnm{Zhang}, \binits{J.}}, et al.:
{End to end learning for self-driving cars}.
arXiv preprint arXiv:1604.07316
(2016)
\end{botherref}
\endbibitem

%%% 13
\bibitem{Mehonic2022Brain-inspiredPlan}
\begin{barticle}
\bauthor{\bsnm{Mehonic}, \binits{A.}},
\bauthor{\bsnm{Kenyon}, \binits{A.J.}}:
\batitle{{Brain-inspired computing needs a master plan}}.
\bjtitle{Nature}
\bvolume{604}(\bissue{7905}),
\bfpage{255}--\blpage{260}
(\byear{2022}).
\doiurl{10.1038/s41586-021-04362-w}
\end{barticle}
\endbibitem

%%% 14
\bibitem{Rhodes2020Brain-inspiredCompleteness}
\begin{botherref}
\oauthor{\bsnm{Rhodes}, \binits{O.}}:
{Brain-inspired computing boosted by new concept of completeness}
(2020).
\doiurl{10.1038/d41586-020-02829-w}
\end{botherref}
\endbibitem

%%% 15
\bibitem{Manin2010AMathematicians}
\begin{bbook}
\bauthor{\bsnm{Manin}, \binits{Y.I.}}:
\bbtitle{{A Course in Mathematical Logic for Mathematicians}}
vol. \bseriesno{53},
(\byear{2010})
\end{bbook}
\endbibitem

%%% 16
\bibitem{Lloyd1996UniversalSimulators}
\begin{botherref}
\oauthor{\bsnm{Lloyd}, \binits{S.}}:
{Universal quantum simulators}.
Science
\textbf{273}(5278)
(1996).
\doiurl{10.1126/science.273.5278.1073}
\end{botherref}
\endbibitem

%%% 17
\bibitem{Sakurai2017ModernNapolitano.}
\begin{bbook}
\bauthor{\bsnm{Sakurai}, \binits{J.J.}}:
\bbtitle{{Modern Quantum Mechanics / J. J. Sakurai, Jim Napolitano.}},
(\byear{2017})
\end{bbook}
\endbibitem

%%% 18
\bibitem{Masanes2019TheRedundant}
\begin{botherref}
\oauthor{\bsnm{Masanes}, \binits{L.}},
\oauthor{\bsnm{Galley}, \binits{T.D.}},
\oauthor{\bsnm{M{\"{u}}ller}, \binits{M.P.}}:
{The measurement postulates of quantum mechanics are operationally redundant}.
Nature Communications
\textbf{10}(1)
(2019).
\doiurl{10.1038/s41467-019-09348-x}
\end{botherref}
\endbibitem

%%% 19
\bibitem{Piccinini2007ComputationalismFallacy}
\begin{botherref}
\oauthor{\bsnm{Piccinini}, \binits{G.}}:
{Computationalism, the church-turing thesis, and the church-turing fallacy}.
Synthese
\textbf{154}(1)
(2007).
\doiurl{10.1007/s11229-005-0194-z}
\end{botherref}
\endbibitem

%%% 20
\bibitem{Kaye2006AnComputing}
\begin{bbook}
\bauthor{\bsnm{Kaye}, \binits{P.}},
\bauthor{\bsnm{Laflamme}, \binits{R.}},
\bauthor{\bsnm{Mosca}, \binits{M.}}:
\bbtitle{{An Introduction to Quantum Computing}},
(\byear{2006}).
\doiurl{10.1093/oso/9780198570004.001.0001}
\end{bbook}
\endbibitem

\end{thebibliography}

\begin{appendices}

\section{Local Write Order}\label{app:write_order}
Up until now, the choice of when to increment the local time might have seemed arbitrary considering the fact that several write operations occur just before we increment the local clock. We provide a more extensive discussion on our choice of convention for the local clock and show that the local machine cannot perceive the write operations occurring separately due to another undecidability result. Here, we prime that discussion by showing how the local machine perceives all updates in between two local time steps $\tau$ and $\tau + 1$ as occurring simultaneously. We can show this is the case by proving it is impossible for the local machine to compute the order in which the write operations occur just before a local time step. To do so, we first denote a write operation as $W(b,s)$ where the symbol $b \in \Gamma$ is written at the location on the tape $S_t$ indexed by $s$.

\begin{problem}[WRITEORDER]
Consider a relative model $\mathcal{P}(M,M')$. We denote the set of global time steps $k_\tau \in K$ where the global machine $M$ updates the local machine $M'$ in accordance with its transition function $\delta'$ as $K = \{k_1, k_2, ... \}$ such that $\delta'(S'_{k_\tau}, q'_{k_\tau}) \rightarrow \{S'_{k_{\tau+1}}, q'_{k_{\tau+1}}\}$. For each global time step $k_\tau \in K$, there exists an ordered set of write operations $\mathcal{W}_{\tau} = \{W(b_1,s_1) ,.., W(b_{\|\mathcal{W}\|}, s_{ \|\mathcal{W} \| } )\}$ required by $M$ to update the tape state and machine state of the local machine $M'$. Construct a local machine $M'$ to compute $\mathcal{W}_{\tau}$ without querying $M$ for the solution.
\end{problem}
\begin{corollary}
WRITEORDER is undecidable. \label{thm:writeorder_und}
\end{corollary}
\begin{proof}
To compute the order of $\mathcal{W}_{\tau}$, the local machine $M'$ needs to perform read operations on $S'_t$ at absolute times between $k_{\tau}$ and $k_{\tau + 1} - 1$, which requires knowledge of the runtime tape set $\Delta_\tau$. Therefore, solving WRITEORDER requires computing a simulation property, which is undecidable by Theorem \ref{thm:SIMPROPERTY_undecidable}.
\end{proof}

One can see by Corollary \ref{thm:writeorder_und} that if the local machine can't perceive the order in which the write operations occur, then from its local frame they all occur simultaneously. This is further discussed in Appendix \ref{app:tau_choice}. This seems to contradict the intuition that time is measured by a change. If the local time was measured by the change in $S'_t$, then the local time of the local machine would change with every write operation. But, it can't perceive the write operations as occurring separately. Therefore, the local frame of the local machine is equivalent to the local time the machine experiences if it were real, namely at the moment the transition function is computed. So, the perception of local time is at every instance of the transition function and only at the transition function.

\section{Local Clock Conventions}
\label{app:tau_choice}
In Section \ref{sec:relative_model}, we choose to increment $\tau$ only when $S_{k_{\tau+1}}'$ and $q_{k_{\tau+1}}'$ are updated according to the transition function of $M'$. In this section, we choose an alternative convention for the local clock as in Appendix \ref{app:write_order} and show that these definitions also cannot compute the specific global time $t$ of an update. In this convention, we define the new local clock $\Tilde{\tau}$ to increment whenever there is a partial or complete update to the local machine's tape state $S'_t$ or machine state $q'_t$. This convention feels more natural in the sense that time should measure all changes. With this alternative clock convention, we still require that the tape state $S'_t$ and machine state $q'_t$ will not change except for when the global machine $M$ writes the updated states. Under the new definition of local time, both computing the time in between local transition function updates and computing the time of write operations are still undecidable.

\begin{problem}[WRITETIME]
Consider a relative model $\mathcal{P}(M,M')$. We denote the set of all global time steps where the global machine $M$ makes any partial or complete update to the local machine's tape state or machine state as $\Tilde{K} = \{\Tilde{k}_1, \Tilde{k}_2, ... \}$. Then, construct a local machine $M'$ to compute $\Tilde{k}_{\Tilde{\tau}+1} - \Tilde{k}_{\Tilde{\tau}}$ for any $\Tilde{k}_{\Tilde{\tau}} \in \Tilde{K}$ without querying the solution from $M$.
\end{problem}
\begin{corollary}
WRITETIME is undecidable. \label{thm:WRITETIME_undecidable}
\end{corollary}
\begin{proof}
Assume for the sake of contradiction that WRITETIME is decided by $R$.  We can use $R$ to build a decider for WRITEORDER by trivially returning the ordering of the list $\Tilde{K}$.  Since we are able to build a decider for WRITEORDER, and WRITEORDER is undecidable by Corollary \ref{thm:writeorder_und}, there is a contradiction, and thus WRITETIME is also undecidable.
% By Corollary \ref{thm:writeorder_und}, computing the order of write operations is undecidable. Because knowing the order of write operations is automatically implied by a list of global times for the write operations, i.e., all $\Tilde{K}$, it must be impossible to also compute $\Tilde{K}$. Therefore, WRITETIME is undecidable because WRITEORDER is undecidable.
%This means it is impossible to even know which bit was written first, let alone at what global time. 
%By Theorem \ref{thm:SIMTIME_undecidable}, $M'$ cannot compute $k_{\tau+1} - k_{\tau}$ using the alternate convention of the local clock. However, for every local time step using the standard convention of local time $\tau$, there exists as many time steps $\Tilde{\tau}$ as there are write operations in the alternative convention $\Tilde{\tau}$. Consider two alternative local time steps $\Tilde{\tau}$ and $\Tilde{\tau}+1$ 
\end{proof}

So, regardless of which convention is chosen for the local clock, the time in between write operations still can't be computed directly. Therefore, there is no benefit to choosing either convention computationally. However, choosing $\tau$ to increment with the local transition function is more in line with the computational abilities of the local machine.

\section{Comparing Relative Models to Oracle Models}
\label{app:oracles}
As shown in Theorem \ref{thm:seminary_problem}, with respect to the local clock $\tau$, a relative model and an oracle model can be equivalent for computable functions. We differentiate a relative model from an oracle model by noting that oracle models bestow more global computational power to the overall model with respect to the global clock $t$. Whereas relative models do not increase the overall computational power with respect to $t$. Additionally, Theorem \ref{thm:seminary_problem} gives us a method of building a local oracle model. However, we cannot build an oracle model unless the analysis is with respect to the local machine.

\section{Probabilistic SIMTIME}
\label{app:pSIMTIME}
We define a probabilistic transition function $\delta$ that maps the current state $S_{t}$ of the tape to another state $S_{t+1}$ by choosing from a fixed set of possible states $\mathcal{S}_\rho(S_{t})$ according to a probability distribution $\rho$.

\begin{problem}[pSIMTIME]
Consider a relative model $\mathcal{P}(M,M')$ where both $\delta$ and $\delta'$ are probabilistic transition functions under probability distributions $\rho$ and $\rho'$, respectively. We denote the set of all global time steps where $M$ updates $M'$ as $K = \{k_1, k_2, ... \}$. Then, construct a local machine $M'$ that computes $k_{\tau+1} - k_\tau$ for any $k_\tau \in K$ without querying $M$.
\end{problem}
\begin{corollary}
pSIMTIME is undecidable. \label{thm:pSIMTIME_undecidable}
\end{corollary}
\begin{proof}

We will prove this by contradiction. Consider two global clock times $k_\tau$ and $k_{\tau+1}$ and assume that $M'$ can compute $k_{\tau+1} - k_\tau$. Then, let $S'_{k_\tau}$ be the state of the local machine's $M'$ tape at global time $k_\tau$ immediately after applying $\delta'$ at $k_\tau$. At global time step $k_{\tau+1}$, the global machine $M$ has simulated $M'$ by computing the probabilistic transition function $\delta'$ that maps $S_{k_\tau}$ to some string $S_{k_{\tau+1}} \in \mathcal{S}_\rho(S_{k_{\tau}})$. It's important to note that $\delta'$ is a function of $S'_{k_\tau}$ to a string $S'_{k_{\tau+1}} \in \mathcal{S}_q(S'_{k_{\tau}})$ that encodes some information that can compute $k_{\tau+1} - k_\tau$. Now, suppose we considered a different interval end time $\bar{k}_{\tau+1}$. Then, the local machine's transition function $\delta'$ constructs a different state $S'_{\bar{k}_{\tau+1}} \in \mathcal{S}'_q(S_{k_{\tau}})$ that encodes $k'_{\tau+1} - k_\tau$ as opposed to the state $S'_{k_{\tau+1}} \in \mathcal{S}_\rho(S'_{k_{\tau}})$ that encodes $k_{\tau+1} - k_\tau$. If it is ever possible for $S'_{\bar{k}_{\tau+1}} = S'_{k_{\tau+1}}$, or at least the information that encodes the time interval difference is the same, then $M'$ cannot differentiate information that encodes $\bar{k}_{\tau+1} - k_\tau$ from information that encodes $k_{\tau+1} - k_\tau$. Therefore, $M'$ must be defined in such a way that it is impossible to map to the same string for different global time intervals, i.e., $\mathcal{S}_\rho(S_{\bar{k}_{\tau}}) \cap \mathcal{S}_\rho(S_{k_{\tau}}) = \emptyset$. However, by imposing this restriction, $\delta'$ can map to two different sets of strings from the same intial state $S_{k_{\tau}}$, which contradicts its definition.
\end{proof}




\end{appendices}




\end{document}